\documentclass[11pt]{article}
\usepackage[letterpaper, left=1in, top=1in, right=1in, bottom=1in,nohead,includefoot]{geometry}
\usepackage[utf8]{inputenc}
\usepackage[authoryear,round]{natbib}
\usepackage{adjustbox}
\usepackage{indentfirst}
\usepackage{mathrsfs}
\usepackage{amssymb}
\usepackage{amsmath}
\usepackage{bbm}
\usepackage{mathtools}
\usepackage{ascmac}
\usepackage{amsthm}
\usepackage{cmll}
\usepackage{stmaryrd}
\usepackage{caption}
\usepackage{lscape}
\usepackage{subcaption}
\usepackage{natbib}
\usepackage{setspace}
\usepackage{bm}
\usepackage{url}
\usepackage{booktabs}
\usepackage{color}
\usepackage{enumerate}
\usepackage{tabularx}

\usepackage{times}

\usepackage[dvipsnames]{xcolor}
	\usepackage[pdftex]{hyperref}
	\hypersetup{colorlinks,
	citecolor=NavyBlue  
	}

\usepackage{subfiles}

\usepackage[colorinlistoftodos,prependcaption,textsize=tiny]{todonotes}

\usepackage{mathrsfs}
\usepackage{amssymb}
\usepackage{amsmath}
\usepackage{mathtools}
\usepackage{ascmac}
\usepackage{amsthm}
\usepackage{multirow}
\usepackage{proba}
\usepackage{subcaption}
\usepackage{adjustbox}

\usepackage{graphicx}
\usepackage{natbib}
\usepackage{setspace}
\usepackage{bm}
\usepackage{url}

\usepackage{color}

\usepackage{times}

\newtheorem{theorem}{Theorem}
\newtheorem{lemma}{Lemma}
\newtheorem{corollary}{Corollary}
\newtheorem{proposition}{Proposition}
\newtheorem{example}{Example}

\newtheorem{assumption}{Assumption}
\newtheorem{remark}{Remark}

\begin{document}

\title{When Is Generalized Bayes Bayesian? \\A Decision-Theoretic Characterization of Loss-Based Updating}
\author{Kenichiro McAlinn\thanks{Department of Statistics, Operations, and Data Science, Fox School of Business, Temple
University, Philadelphia, PA. {\scriptsize{}Email:
kenichiro.mcalinn@temple.edu}}\, \& K\={o}saku Takanashi\thanks{RIKEN Center for Advanced Intelligence Project {\scriptsize{}Email: kosaku.takanashi@riken.jp} }}

\maketitle
\thispagestyle{empty}\setcounter{page}{0}
\begin{abstract}
Loss-based updating, including generalized Bayes, Gibbs, and quasi-posteriors, replaces
likelihoods by a user-chosen loss and produces a posterior-like distribution via exponential
tilt. We give a decision-theoretic characterization that separates
\emph{belief posteriors}-- conditional beliefs justified by the foundations of Savage and Anscombe--Aumann under a joint
probability model-- from \emph{decision posteriors}-- randomized decision rules justified by
preferences over decision rules.
We make explicit that a loss-based posterior coincides with ordinary Bayes if and only if the loss is, up to scale and a
data-only term, negative log-likelihood. We then show that generalized marginal likelihood is not evidence for decision
posteriors, and Bayes factors are not well-defined without additional structure.
In the decision posterior regime, non-degenerate posteriors require
nonlinear preferences over decision rules. Under sequential coherence and separability, these lead
to an entropy-penalized variational representation yielding generalized Bayes as the optimal rule.

\bigskip{}
{\it Keywords}: generalized Bayes, Gibbs posterior, loss-based updating, decision theory, dynamic consistency, marginal likelihood, model comparison, learning-rate calibration
\end{abstract}
\newpage{}

\section{Introduction}\label{sec:intro}

A large and growing literature updates a baseline distribution by exponentiating cumulative loss,
producing what is routinely called a ``generalized posterior.''  This loss-based updating viewpoint
includes general Bayesian inference and generalized Bayes \citep{BissiriHolmesWalker2016GeneralFramework},
Gibbs and quasi-posteriors, and related exponential-weights constructions.  It is widely used in robust
and misspecified inference \citep{BissiriHolmesWalker2016GeneralFramework,LyddonHolmesWalker2019LossLikelihoodBootstrap,
GrunwaldVanOmmen2017InconsistencyMisspecifiedLinearModels,WatsonHolmes2016ApproximateModelsRobustDecisions},
econometrics based on criterion functions \citep{chernozhukov2003mcmc}, learning-theoretic analyses
\citep{catoni2007pac}, decision synthesis \citep{tallman2024bayesian}, and inverse problems where a likelihood is unavailable or unreliable.
These methods are practically appealing, since the losses can be easier to elicit than sampling models, align directly with downstream goals,
and offer robustness when likelihood models are untenable.

However, calling the output a ``posterior'' invites an ambiguity that matters operationally.
Bayesian posteriors carry a specific coherence story. They are conditional beliefs derived from a joint probability model,
supporting familiar interpretations (betting odds, exchangeability, evidential use of marginal likelihoods, and Bayes factors).
Loss-based updates are often discussed using this same belief language even when no sampling model is specified. {Terminology in this area is not standardized, where loss-based/Gibbs updates are sometimes described using the language of belief updating or posterior credences, or “genuine posteriors,” even when the update is driven by a user-chosen loss}
Thus, we ask, which parts of the Bayesian coherence story survive once likelihood is replaced by loss?

This paper provides a decision-theoretic traffic control for loss-based updating by separating two objects that are often conflated.
A \emph{belief posterior} is a conditional belief about an uncertain state, derived from preferences over acts in the
Savage/Anscombe--Aumann framework \citep{Savage1954Foundations,AnscombeAumann1963DefinitionSubjectiveProbability}.
Under conditional choice axioms, a dynamic consistency requirement links unconditional and conditional preferences and
selects Bayesian conditioning as the dynamically consistent update \citep{Ghirardato2002RevisitingSavage}.
A \emph{decision posterior}, by contrast, is a randomized decision rule: a data-dependent distribution over actions
(or parameter indices) chosen to optimize a specified preference functional after observing data.
Both objects can be coherent-- but in different senses.
Belief posteriors are coherent as conditional beliefs; decision posteriors are coherent as optimal randomization devices.
Our central message is that loss-based updating is naturally aligned with decision posteriors and coincides with belief posteriors
only in special cases.  Making this explicit clarifies what ``priors'' mean in generalized Bayes, what form of dynamic consistency
is (and is not) being enforced, why a learning-rate parameter is unavoidable, and which Bayesian properties
(e.g., exchangeability as a belief constraint; likelihood-principle rhetoric; evidential use of marginal likelihoods) are not available
without additional assumptions.

This distinction has operational implications, and is most prevalent when we consider model comparison.
In the belief-posterior regime, the marginal likelihood is a prior predictive probability under a normalized sampling model
and supports Bayes-factor interpretations.
Outside that regime, there is no such canonical evidential object.
Here ``canonical'' means invariant to transformations of the loss that leave the posterior decision rule unchanged.
Loss-based posteriors are invariant to adding an arbitrary data-only term to the loss; however the Gibbs normalizing
constant rescales.  Bayes-factor-style comparison based on generalized marginal likelihood is not well-defined
without additional structure.

\paragraph{Contributions.}
We make four contributions aimed at separating what is genuinely Bayesian belief updating from what is decision-theoretic regularization.
First, we give a sharp divider: a loss-based posterior is an ordinary Bayesian \emph{belief posterior}
if and only if the loss is, up to a data-only term and scale, negative log-likelihood.
Second, we show that generalized marginal likelihood is not a \emph{canonical} evidential quantity outside the belief regime:
the posterior decision rule is invariant to data-only shifts of the loss, while the Gibbs normalizing constant rescales,
so Bayes-factor rhetoric requires additional identifying structure; we also give a shift-invariant alternative based on
proper predictive scoring rules.
Third, we formalize a degeneracy implication of von~Neumann--Morgenstern expected utility for mixtures of decision rules:
optimal randomization collapses to point-mass choices, implying that any non-degenerate generalized posterior necessarily relies on
nonlinear preferences over decision rules.
Fourth, we provide a preference-based foundation for generalized Bayes as a \emph{decision posterior}.
Building on variational representations from economics and information theory, we show how sequential coherence
together with separability across independent subproblems selects an entropy-based deviation cost,
yielding the generalized Bayes form as an optimal randomized decision rule; read generatively, this also provides a
simple design recipe for constructing loss-based updates and diagnosing when KL-regularized generalized Bayes is not canonical.

Beyond these results, we provide a ``coherence book'' itemizing which claims are valid for belief posteriors and which fail
for decision posteriors.  In particular, calibration of the learning rate $\eta$ moves
from a modeling detail (fixed by likelihood scale) to a structural component of principled analysis,
clarifying why calibration procedures such as the loss-likelihood bootstrap \citep{LyddonHolmesWalker2019LossLikelihoodBootstrap}
and safe learning-rate selection \citep{GrunwaldVanOmmen2017InconsistencyMisspecifiedLinearModels} are requirements once one leaves the belief-posterior regime.

\paragraph{Outline.}
Section~\ref{sec:setup} formalizes belief versus decision posteriors and introduces primitives for loss-based updating.
Section~\ref{sec:characterize} gives the Bayesian equivalence divider and the preference characterization leading to generalized Bayes.
Section~\ref{sec:ml} develops the model-comparison consequence.
Section~\ref{sec:vnm} records the von~Neumann--Morgenstern degeneracy implication.
Section~\ref{sec:book} presents the coherence book and discusses practical implications.

\section{Belief posteriors and decision posteriors}\label{sec:setup}

This section fixes primitives and terminology.  Our goal is not to re-axiomatize Bayesian decision
theory, but to make explicit a distinction that becomes essential once one replaces likelihood with
loss: whether the ``posterior'' object is (i) a conditional belief state about uncertain outcomes,
or (ii) a randomized decision rule chosen after observing data.

\subsection{Belief posteriors: conditional beliefs from preferences over uncertain acts}\label{subsec:belief}

In the Savage/Anscombe--Aumann tradition, the primitive is a preference relation over (possibly
randomized) acts under uncertainty \citep{Savage1954Foundations,AnscombeAumann1963DefinitionSubjectiveProbability}.
An act maps an uncertain state to a consequence, and preferences are required to be
compatible with conditional choice.  Under standard axioms, there exist a utility function over
consequences and a subjective probability over states such that acts are ranked by subjective
expected utility.  Dynamic consistency links unconditional and conditional rankings: plans chosen ex
ante must agree with choices made ex post when the relevant information event is learned.  In
particular, Bayesian conditioning arises as the unique dynamically consistent update \citep{Ghirardato2002RevisitingSavage}.

In statistical applications, belief posteriors arise when one posits a joint probabilistic model for
data and parameters.  A prior distribution represents beliefs about the parameter, a likelihood represents beliefs about observable consequences given that
state, and Bayesian conditioning produces a posterior distribution.  This posterior is a belief
object: it is interpretable as a coherent conditional belief state.

\subsection{Decision posteriors: post-data randomized decision rules}\label{subsec:decision}

Loss-based updating is naturally expressed in a different primitive.  Fix an action (or parameter)
space $\Theta$; we use ``parameter'' only as an index for a family of candidate actions/predictors.
After observing data $x$, the decision maker chooses a randomized decision rule
$q(\cdot| x)\in\Delta(\Theta)$, where $\Delta(\Theta)$ denotes probability measures on $\Theta$.
The key point is conceptual: $q(\cdot| x)$ is a choice variable, not necessarily a belief
about an underlying data-generating state.

To evaluate randomized rules, we assume the analyst specifies a loss function $\ell(\theta;x)$ that
scores the action $\theta$ on the realized data $x$.  For repeated or sequential data
$x_{1:n}=(x_1,\ldots,x_n)$, we write the cumulative loss as
\[
L_n(\theta) = \sum_{i=1}^n \ell(\theta;x_i),
\]
and treat $L_n$ as the data-dependent performance criterion.  A baseline
distribution $\pi$ on $\Theta$ serves as a reference rule, encoding default behavior, complexity
regularization, or pre-data commitments; crucially, it need not be a prior belief about outcomes.

A \emph{decision posterior} is then any rule $q(\cdot| x_{1:n})$ that is selected by optimizing a
preference functional defined on $\Delta(\Theta)$.  The generalized Bayes family corresponds to the
Gibbs form
\begin{equation}\label{eq:gibbs}
q(\theta| x_{1:n}) \;\propto\; \pi(\theta)\exp\{-\eta L_n(\theta)\},
\end{equation}
for some learning rate $\eta>0$.  This includes ordinary Bayes as a special case
when the loss is proportional to negative log-likelihood, but in general the interpretation is
different: \eqref{eq:gibbs} is a randomized decision rule derived from a loss and a baseline.

\subsection{Loss-based updating as an operator and sequential coherence}\label{subsec:operator}

It is useful to view loss-based updating as an operator that maps primitives to a rule on
$\Delta(\Theta)$.  Let $\Psi$ be an updating map such that, for data $x_{1:n}$,
\[
q(\cdot| x_{1:n}) = \Psi\!\left(\pi,\;L_n(\cdot)\right).
\]
A central desideratum is sequential (batching) coherence:
updating with $x_1$ and then $x_2$ should agree with updating once using the combined loss:
\[
\Psi\!\left(\Psi(\pi, L_m),\; L_{m+1:n}\right) = \Psi(\pi, L_n).
\]
The Gibbs form \eqref{eq:gibbs} satisfies this by construction \citep{BissiriHolmesWalker2016GeneralFramework}.

A crucial difference from belief posteriors is scale.  In ordinary Bayes, the scale of
negative log-likelihood is fixed by the probability model, and the corresponding ``temperature'' is
one.  In loss-based updating, $\eta$ is not pinned down by probabilistic semantics; it
is part of the decision specification.  This is why calibration of $\eta$ is structural rather than
cosmetic.

\subsection{Standing assumptions and scope}\label{subsec:assumptions}

Throughout, $\Theta$ is a measurable space and $\pi$ is a probability measure on $\Theta$.  The loss
$\ell(\theta;x)$ is measurable in $\theta$ for each $x$ and such that $L_n(\theta)$ is well-defined
for the settings of interest.  Our results concern the interpretation and decision-theoretic
status of distributions produced by loss-based updates: when they can be interpreted as belief
posteriors, when they must be interpreted as decision posteriors, and what coherence properties are
retained or lost under each interpretation.

\section{When is generalized Bayes Bayesian?}\label{sec:characterize}

This section gives two characterizations that organize the remainder of the paper.
First, we formalize when a loss-based posterior can be reinterpreted as an ordinary Bayesian
belief posterior under some likelihood model.  Second, we give a decision-theoretic
foundation for generalized Bayes as a decision posterior: an optimal randomized rule under
sequential coherence and a separability invariance across independent subproblems.

\subsection{Generalized Bayes coincides with Bayes if and only if the loss is log-likelihood}\label{subsec:iff}

Fix a baseline distribution $\pi$ on $\Theta$ and a loss $\ell(\theta;x)$.  For a single
observation $x$ (or for $x$ denoting a dataset), write the generalized Bayes update as
\begin{equation}\label{eq:gb-update}
q(\theta| x)\;\propto\;\pi(\theta)\exp\{-\eta \ell(\theta;x)\},
\end{equation}
with learning rate $\eta>0$.  We ask: when can $q(\cdot| x)$ be interpreted as a Bayesian
posterior arising from a joint model and thus inherit belief-posterior semantics?

\begin{theorem}[Bayesian equivalence]\label{thm:bayes-iff}
Fix $\eta>0$, a baseline distribution $\pi$ on $\Theta$, and a measurable loss $\ell(\theta;x)$.
The following are equivalent:
\begin{enumerate}[(i)]
\item There exists a likelihood family $\{p_\theta(x)\}_{\theta\in\Theta}$ such that for all $x$,
the generalized update \eqref{eq:gb-update} equals the ordinary Bayesian posterior,
$q(\theta| x)\propto \pi(\theta)p_\theta(x)$.
\item There exist a likelihood family $\{p_\theta(x)\}_{\theta\in\Theta}$ and a function $c(x)$ not
depending on $\theta$, such that, for all $\theta$ and $x$,
\begin{equation}\label{eq:loss-loglik-affine}
\ell(\theta;x) = -\frac{1}{\eta}\log p_\theta(x) \;+\; c(x).
\end{equation}
\end{enumerate}
\end{theorem}

Although variants of Theorem~\ref{thm:bayes-iff} appear in the literature \citep[see, e.g.,][]{BissiriHolmesWalker2016GeneralFramework,catoni2007pac,GrunwaldVanOmmen2017InconsistencyMisspecifiedLinearModels}, the above formulation provides a sharp divider in our context.  A generalized posterior is a belief
posterior-- it can be interpreted as Bayesian conditioning under some joint model-- if and
only if the loss is, up to an irrelevant data-only term and scale, negative log-likelihood.  In all
other cases, \eqref{eq:gb-update} is not (and cannot be made) Bayesian conditioning in the
Savage/Anscombe--Aumann sense: it must be interpreted as a decision posterior. Proof is in Supplementary Material~\ref{app:bayes-iff}.

\begin{example}[Criterion-function loss breaks belief semantics]
    For $y_{1:n}$ and a scale parameter $\sigma>0$, consider the loss
$L_n(\sigma)=\sum_{i=1}^n y_i^2/(2\sigma^2)$, yielding
$q(\sigma| y)\propto \pi(\sigma)\exp\{-\eta L_n(\sigma)\}$.
Under the Gaussian sampling model, the negative log-likelihood differs from $L_n(\sigma)$ by the
parameter-dependent term $n\log\sigma$, so Theorem~\ref{thm:bayes-iff} implies this update is not a
belief posterior. It must be interpreted as a decision posterior.
\end{example}

So, why not ``normalize $\exp\{-\eta\ell\}$?''
A natural objection is that for any loss one can define a formal family
$\tilde p_\theta(x)\propto \exp\{-\eta \ell(\theta;x)\}$ by normalizing over $x$.  Theorem~\ref{thm:bayes-iff}
shows why this does \emph{not} generally yield a belief-posterior interpretation: the normalizing
factor needed to make $\tilde p_\theta$ integrate to one typically depends on $\theta$ and
therefore cannot be absorbed into $c(x)$ in \eqref{eq:loss-loglik-affine}.  Being ``Gibbs in $\theta$'' is automatic, but being a likelihood in $x$ with a
$\theta$-free additive constant is a substantive restriction.

\begin{corollary}[$x$-normalization diagnostic for Bayesianity]\label{cor:bayes-diagnostic}
Assume the sample space $\mathcal X$ carries a dominating measure $\mu$.
Fix $\eta>0$ and a measurable loss $\ell(\theta;x)$.  Define
\[
A(\theta)=\int_{\mathcal X}\exp\{-\eta \ell(\theta;x)\} \mu(dx)\in(0,\infty].
\]
If there exists a likelihood family $\{p_\theta\}$, such that
$q(\theta| x)\propto \pi(\theta)p_\theta(x)$ agrees with \eqref{eq:gb-update} for all $x$, then
$A(\theta)$ must be finite and independent of $\theta$.
Conversely, if $A(\theta)\equiv A\in(0,\infty)$, then
$p_\theta(x)=\exp\{-\eta \ell(\theta;x)\}/A$
defines a valid likelihood and \eqref{eq:gb-update} is an ordinary Bayesian posterior.
\end{corollary}

Corollary~\ref{cor:bayes-diagnostic} makes explicit the substantive content of
Theorem~\ref{thm:bayes-iff}: generalized Bayes is a belief posterior only when the ``partition
function'' required to normalize $\exp\{-\eta\ell(\theta;x)\}$ as a density in $x$ does \emph{not}
depend on $\theta$.

\subsection{From preferences over decision rules to generalized Bayes}\label{subsec:variational}

We now turn to a positive characterization: when is the generalized Bayes form \eqref{eq:gb-update}
the optimal randomized decision rule derived from an axiomatic preference specification?
All proofs are in Supplementary Material~\ref{app:variational}.

After observing data $x$, the decision maker chooses a randomized rule $q\in\Delta(\Theta)$.  The
data enter through the loss $\ell(\theta;x)$, and the baseline $\pi$ represents a default rule or
regularizer.  We seek preference conditions under which the optimal $q$ must take the generalized
Bayes form.

A natural class of decision criteria on $\Delta(\Theta)$ is variational: the decision maker
trades off expected loss under $q$ against a penalty for deviating from $\pi$,
\begin{equation}\label{eq:variational-form}
q^{*}(\cdot| x)\;\in\;\arg\min_{q\in\Delta(\Theta)}
\left\{
\int \ell(\theta;x)\, q(d\theta) \;+\; \frac{1}{\eta} D(q\|\pi)
\right\},
\end{equation}
where $D(\cdot\|\cdot)$ is a deviation cost and $\eta>0$ calibrates its strength.
This connects to entropy-penalized choice in economics, including variational
and multiplier preferences \citep{HansenSargent2001,MaccheroniMarinacciRustichini2006,Strzalecki2011} and information-cost
foundations for stochastic choice \citep{MatejkaMcKay2015}.

\paragraph{Axioms (variational decision posteriors).}
Fix $x$ and let $\succeq_x$ denote conditional preferences over $q\in\Delta(\Theta)$.
We impose the following structure:
\begin{enumerate}[({A}1)]
\item \emph{Weak order and regularity:} $\succeq_x$ is complete and transitive and is continuous
under mixtures.
\item \emph{Loss monotonicity:} if $\int \ell(\theta;x)\, q(d\theta)\le \int \ell(\theta;x)\, r(d\theta)$,
then $q\succeq_x r$, with strict preference when the inequality is strict and all else is equal.
\item \emph{Sequential (batching) coherence:} for independent data blocks $x$ and $x'$ whose losses
add, preferences induced after observing $(x,x')$ coincide with those induced by updating sequentially.
\item \emph{Separability across independent subproblems:} if the action space and baseline factor as
$\Theta=\Theta_1\times\Theta_2$ and $\pi=\pi_1\otimes\pi_2$, and if losses add
$\ell((\theta_1,\theta_2);x)=\ell_1(\theta_1;x)+\ell_2(\theta_2;x)$, then the induced optimal rule
factors: $q^{*} = q_1^{*}\otimes q_2^{*}$, and the preference ranking of product rules depends only on
their component rankings.
\end{enumerate}

The separability axiom is the key ``no interaction across independent subproblems'' requirement:
it rules out deviation penalties that introduce spurious coupling when the baseline and losses do
not.  This is the decision-theoretic analogue of product/recursivity conditions that underlie
entropy-based penalties in robust choice \citep{Strzalecki2011,MatejkaMcKay2015}.

\begin{assumption}[Variational representation]\label{ass:variational}
For each $x$, there exist $\eta>0$ and a functional $D(\cdot\|\pi):\Delta(\Theta)\to[0,\infty]$, such
that preferences admit the variational form \eqref{eq:variational-form}.  We further assume
$D(\pi\|\pi)=0$ and that $D(\cdot\|\pi)$ is strictly convex and lower semicontinuous in $q$.
\end{assumption}

\begin{assumption}[$f$-divergence deviation cost]\label{ass:fdiv}
For each baseline $\pi$, the deviation cost admits the separable form
$D(q\|\pi)=\int_{\Theta}\varphi(dq/d\pi)\,d\pi$
whenever $q\ll \pi$, and $D(q\|\pi)=+\infty$ otherwise, where $\varphi:(0,\infty)\to\mathbb R$ is
convex with $\varphi(1)=0$.  We assume strict convexity in $q$.
\end{assumption}

Assumption~\ref{ass:fdiv} covers the most widely used penalties in decision theory and statistics
(including relative entropy), while remaining broad enough that the identification result below is nontrivial.

\begin{proposition}[Separability implies product additivity]\label{prop:prodadd}
Assume Axiom (A4) and Assumption~\ref{ass:variational}.
Then, for product spaces $\Theta=\Theta_1\times\Theta_2$ and product baselines
$\pi=\pi_1\otimes\pi_2$, separability implies the product-additivity restriction
\[
D(q_1\otimes q_2\|\pi_1\otimes\pi_2)
=
D(q_1\|\pi_1)+D(q_2\|\pi_2),
\quad\text{for all } q_1\in\Delta(\Theta_1),\ q_2\in\Delta(\Theta_2).
\]
\end{proposition}

\begin{proposition}[KL is the unique $f$-divergence with product additivity]\label{prop:klunique}
Assume Assumption~\ref{ass:fdiv}.  If $D$ satisfies the product-additivity property of
Proposition~\ref{prop:prodadd} for all product baselines and product measures, then
$D(q\|\pi)=c\,D_{\mathrm{KL}}(q\|\pi)$ for $c>0$,
where $c$ can be absorbed into $\eta$ in \eqref{eq:variational-form}.
\end{proposition}

Proposition~\ref{prop:klunique} is an identification result within the broad family of $f$-divergences: product additivity singles out relative entropy (up to scale).

\begin{corollary}[Generalized Bayes as the unique optimal decision posterior]\label{cor:gb}
Assume Axioms (A1)--(A4) and Assumptions~\ref{ass:variational}--\ref{ass:fdiv}.  Then, the unique
optimizer of \eqref{eq:variational-form} satisfies the generalized Bayes form \eqref{eq:gb-update}.
\end{corollary}

The corollary provides a preference-based justification for generalized Bayes as a decision
posterior.  The baseline $\pi$ is interpreted as a reference or default rule and the entropy
penalty as the cost of deviating from that default in a manner that respects separability across
independent subproblems.  The scale $\eta$ becomes a feature of the preference specification rather
than an artifact of a likelihood model, which is why its calibration is central in applications.

Theorem~\ref{thm:bayes-iff} and Corollary~\ref{cor:gb} deliver the ``full picture.''
Theorem~\ref{thm:bayes-iff} identifies when generalized Bayes is Bayesian in the classical sense of
belief updating.  Outside this case,
the appropriate interpretation is Corollary~\ref{cor:gb}: generalized Bayes arises as an
optimal randomized decision rule under a variational preference specification.  Under
our assumptions, the deviation penalty is (up to scale)
relative entropy, yielding the generalized Bayes (Gibbs) rule.  

\subsection{A design recipe for loss-based updating}\label{sec:recipe}

The results above can be read \emph{generatively}: they specify how to construct an update rule once the analyst chooses
the operational notion of fit and the coherence requirements the update should satisfy.

\paragraph{Step 1: Specify the target and loss.}
Choose an action/parameter space $\Theta$ and a loss $\ell(\theta;x)$ encoding the inferential target and the unit of fit.
This choice fixes what is being optimized (risk) and, therefore, what the resulting ``posterior'' weights are meant to accomplish.

\paragraph{Step 2: Specify the composition structure}
Decide whether the problem admits a defensible decomposition into independent subproblems and whether one demands additivity of
both loss and regularization across such components (Axiom~A4).  This is not a technicality. It determines whether the deviation
cost should compose additively.

\paragraph{Step 3: Choose an update family and calibrate temperature}
Within the broad class of entropy-regularized preference representations (Assumption~1), Axiom~A4 identifies KL (Proposition~\ref{prop:klunique}),
yielding the Gibbs/generalized Bayes update.  Outside that regime, KL-regularized generalized Bayes is a modeling choice rather
than a canonical consequence, and the learning rate/temperature $\eta$ must be specified by an explicit calibration rule
(e.g., safe-Bayes, loss-likelihood bootstrap, or held-out predictive matching; Section~\ref{sec:book}).

\paragraph{Diagnostic (when KL is not justified).}
If there is no credible independence structure under which the objective decomposes additively-- for example, in the presence of
global constraints, ranking losses, dependent time series, network interactions, or other couplings-- then Axiom~A4 is not a reasonable
requirement and the KL penalty is not singled out by coherence.  In such cases, the appropriate update depends on the intended
composition rule, and alternative regularizers may be warranted.

In the belief-posterior regime, posterior uncertainty reflects uncertainty about an underlying sampling model
$p(\cdot|\theta)$ together with a prior $\pi(\theta)$, and loss enters only at the decision stage.
Loss-based updating relocates part of this epistemic burden: absent a joint sampling model, the update is driven by
(i) the choice of loss $\ell(\theta;x)$ encoding the operative notion of fit, (ii) the composition/separability structure
invoked to justify a particular deviation penalty (e.g., KL under Axiom~A4), and (iii) the calibration of the learning
rate/temperature $\eta$, whose scale is fixed by the likelihood in belief Bayes but must be chosen explicitly here.
Thus, generalized Bayes should not be seen as ``Bayes without a model,'' rather as ``a model of how you choose under an explicitly specified loss,'' and therefore requires transparency about loss choice and calibration.

\section{Marginal likelihood is not canonical evidence for decision posteriors}\label{sec:ml}

A persistent source of confusion concerns marginal
likelihood and Bayes factor rhetoric.  Under a belief-posterior interpretation, the marginal
likelihood is a well-defined probability of the observed data under a joint model and 
admits a standard evidential interpretation.
For decision posteriors, the quantity that plays the role of a
``normalizing constant'' in a Gibbs posterior is not identified by the update, and
Bayes factor-style evidence is not well-defined without additional structure.

\subsection{Belief evidence versus decision evidence}\label{subsec:ml-belief-vs-decision}

\paragraph{Belief posteriors.}
If a posterior arises from a joint model $p_\theta(x)$ and a prior $\pi(\theta)$, then the marginal
likelihood $m(x)=\int p_\theta(x) \pi(\theta) d\theta$
is the predictive probability of the observed data under the model, and Bayes factors compare evidential support in the standard sense.

\paragraph{Decision posteriors}
For a generalized Bayes (Gibbs) posterior based on loss $\ell(\theta;x)$,
\begin{equation}\label{eq:gibbs-posterior}
q(\theta| x)\;\propto\;\pi(\theta)\exp\{-\eta \ell(\theta;x)\},
\end{equation}
the normalizing constant $Z(x)=\int \exp\{-\eta \ell(\theta;x)\} \pi(d\theta)$
is sometimes used in a Bayes factor-like manner.  However, unless \eqref{eq:gibbs-posterior} is a belief posterior,
$Z(x)$ does not have the belief-evidence interpretation and is not an
object associated with the decision posterior itself.

\paragraph{Canonicality criterion.}
Two losses $\ell$ and $\tilde \ell$ that differ by a data-only term, $\tilde \ell(\theta;x)=\ell(\theta;x)+c(x)$,
induce the same decision posterior via \eqref{eq:gibbs} for any fixed $(\eta,\pi)$.
Any quantity proposed as evidence associated with the decision posterior
should be invariant to such loss-equivalent transformations; we call such an evidential quantity
canonical for the decision posterior.

\subsection{Non-identifiability of generalized marginal likelihood}\label{subsec:ml-nonid}

Generalized Bayes posteriors are invariant to data-only shifts of the
loss, whereas the normalizing constant is not.

\begin{proposition}[Normalizing constants are not identified by a decision posterior]\label{prop:Z-not-identified}
Fix $\eta>0$ and a baseline distribution $\pi$ on $\Theta$.  Let $\ell(\theta;x)$ be a loss and
define the generalized posterior $q(\theta| x)$ by \eqref{eq:gibbs-posterior}, with normalizing
constant $Z(x)$.  For any measurable function $c(x)$, define a shifted loss
$\ell_c(\theta;x)=\ell(\theta;x)+c(x)$.
Let $q_c(\theta| x)$ and $Z_c(x)$ denote the corresponding posterior and normalizing constant.
Then, for all $x$,
\[
q_c(\theta| x)=q(\theta| x),
\quad\text{but}\quad
Z_c(x)=e^{-\eta c(x)} Z(x).
\]
The posterior mapping $x\mapsto q(\cdot| x)$ does not identify $x\mapsto Z(x)$.
\end{proposition}

\begin{proof}
The factor $\exp\{-\eta c(x)\}$ cancels in normalization over $\theta$ but remains in the integral defining $Z$.
\end{proof}

\begin{example}[Empirical likelihood/exponential tilting conventions leave $q$ unchanged but rescale $Z$.]\label{ex:gmm_centering}
Quasi-posteriors are often built by exponentiating empirical-likelihood-type criteria that are only defined up to
data-only constants; see, e.g., \citet{owen2001} and Bayesian variants in
\citet{lazar2003,schennach2005}. Let $g(\theta;x_i)\in\mathbb{R}^k$ be moment conditions and let $\{p_i(\theta)\}_{i=1}^n$
denote the empirical likelihood (EL) weights solving
\[
\max_{p_1,\ldots,p_n}\ \sum_{i=1}^n \log p_i
\quad \text{s.t}\quad \sum_{i=1}^n p_i=1,\ \sum_{i=1}^n p_i g(\theta;x_i)=0,\ p_i\ge 0.
\]
A standard criterion is the (profile) empirical likelihood ratio \citep{owen2001}
\[
\log R(\theta;x)=\sum_{i=1}^n \log\{n\,p_i(\theta)\},
\]
and a common quasi-posterior uses the loss $L(\theta;x)=-\log R(\theta;x)$ (or a scaled version such as $-2\log R$).
Another equally standard convention is to work with the un-ratioed EL objective $-\sum_{i=1}^n \log p_i(\theta)$,
which differs from $L(\theta;x)$ only by the data-only constant $n\log n$:
\[
-\sum_{i=1}^n \log p_i(\theta)=L(\theta;x)+n\log n.
\]

Similarly, exponential tilting (ET) chooses weights $\{w_i(\theta)\}_{i=1}^n$ by minimizing a KL divergence from
the uniform weights subject to the same moment constraints \citep{kitamura_stutzer1997,schennach2005}:
\[
\min_{w_1,\ldots,w_n}\ \sum_{i=1}^n w_i \log\!\left(\frac{w_i}{1/n}\right)
\quad \text{s.t}\quad \sum_{i=1}^n w_i=1,\ \sum_{i=1}^n w_i g(\theta;x_i)=0,\ w_i\ge 0.
\]
Because $\sum_i w_i=1$, the ET objective satisfies
\[
\sum_{i=1}^n w_i \log\!\left(\frac{w_i}{1/n}\right)=\sum_{i=1}^n w_i \log w_i+\log n,
\]
so alternative but equivalent ET conventions differ by the data-only constant $\log n$.

Define a generalized posterior $q(\theta|x)\propto \exp\{-\eta L(\theta;x)\}\,\pi(\theta)$ with normalizing
constant $Z(x)=\int \exp\{-\eta L(\theta;x)\}\pi(\theta)\,d\theta$. If one switches between any two such EL/ET
conventions (or, more generally, centers the criterion for numerical stability), replacing $L$ by
$L^{c}(\theta;x)=L(\theta;x)-c(x)$ for any data-only constant $c(x)$, then $q^{c}(\cdot|x)=q(\cdot|x)$ but
\[
Z^{c}(x)=\int \exp\{-\eta L^{c}(\theta;x)\}\pi(\theta)\,d\theta
= e^{\eta c(x)} Z(x).
\]
(This is Proposition~4 with the identification $c_{\text{Prop.\,4}}(x)=-c(x)$.)
Thus Bayes-factor-like ratios based on generalized marginal likelihood are non-canonical: for two conventions
indexed by $m=1,2$,
\[
\frac{Z_1^{c}(x)}{Z_2^{c}(x)}
=\exp\{\eta(c_1(x)-c_2(x))\}\frac{Z_1(x)}{Z_2(x)},
\]
even though the decision posteriors $q_1(\cdot|x)$ and $q_2(\cdot|x)$ are unchanged.
\end{example}

The same non-canonicity arises from loss scaling.  If $\ell'(\theta;x)=a\,\ell(\theta;x)$ with $a>0$ and
$\eta'=\eta/a$, then the generalized posterior mapping is unchanged since
$\exp\{-\eta'\ell'(\theta;x)\}=\exp\{-\eta\ell(\theta;x)\}$, hence $q'(\cdot|x)=q(\cdot|x)$.
But the normalizing constant depends on the absolute scale of $(\eta,\ell)$ and therefore cannot be
interpreted as intrinsic belief evidence unless the loss scale (equivalently, the temperature) is fixed
by additional structure or external calibration.

\begin{corollary}[Generalized Bayes factors are not well-defined evidence]\label{cor:gbf-not-well-defined}
Consider two loss-based updates indexed by $j\in\{0,1\}$.
Suppose one defines $BF_{10}(x)=Z_1(x)/Z_0(x)$.  Then for any strictly
positive $g(x)$ there exist data-only shifts such that
posteriors are unchanged, $q'_j(\cdot| x)=q_j(\cdot| x)$, but $BF'_{10}(x)=g(x) BF_{10}(x)$.
Outside the belief-posterior regime, $BF_{10}(x)$ is not an evidential quantity determined by the updating rule.
\end{corollary}

Corollary~\ref{cor:gbf-not-well-defined} rules out treating the Gibbs normalizing constant as \emph{canonical belief evidence intrinsic to the
decision posterior mapping}, since it varies under conventions (e.g., data-only shifts) that leave $q(\cdot|x)$
unchanged.  This does not preclude using $Z_m(x)$ as a model-dependent normalizer once a particular loss scale and
convention are fixed; it says that Bayes-factor-style evidence claims require additional identifying structure beyond
the generalized posterior alone.

The following is a simple example on how  evidence claims can flip without changing the posterior.
Consider two analysts who implement the same generalized posterior mapping $q_m(\cdot|x)$ for each model $m$
but use different numerically equivalent stabilizations, i.e., losses that differ by data-only constants $c_m(x)$.
Their reported ``generalized Bayes factor'' ratios based on $Z_m(x)$ then differ by the multiplicative factor
$\exp\{\eta(c_1(x)-c_2(x))\}$ (Example~\ref{ex:gmm_centering}), which can be enormous at realistic sample sizes.
Thus, two implementations can yield identical posteriors and predictions while reporting incompatible evidence conclusions,
underscoring that $Z_m(x)$ is not canonical belief evidence outside the log-loss regime.

A natural response to Proposition~\ref{prop:Z-not-identified} is to select a canonical representative from the shift-equivalence class
$\ell(\theta;x)\sim \ell(\theta;x)+c(x)$ by \emph{anchoring} the loss, e.g.,
\[
\tilde \ell(\theta;x)=\ell(\theta;x)-\inf_{\vartheta\in\Theta}\ell(\vartheta;x)\ge 0,
\qquad
\tilde Z(x)=\int \exp\{-\eta \tilde \ell(\theta;x)\}\,\pi(d\theta)\in(0,1].
\]
Anchoring removes the data-only shift ambiguity: if $\ell'(\theta;x)=\ell(\theta;x)+c(x)$ then $\tilde \ell'=\tilde \ell$ and hence
$\tilde Z'(x)=\tilde Z(x)$. Moreover, by Lemma~1,
\[
-\frac1\eta\log \tilde Z(x)
=\min_{q\in\Delta(\Theta)}\left\{\int \tilde \ell(\theta;x)\,q(d\theta)+\frac1\eta D_{\mathrm{KL}}(q\|\pi)\right\},
\]
so $-\frac1\eta\log \tilde Z(x)$ admits a finite-sample ``complexity/regret'' interpretation relative to the in-sample optimum.

However, anchoring does \emph{not} restore a Bayes-factor analogue for model comparison.  Because the anchor is model-specific,
$\tilde \ell_m$ removes the in-sample fit term $\inf_{\theta}\ell_m(\theta;x)$, so $\tilde Z_m(x)$ measures concentration/complexity
\emph{around} the best-fitting action within model $m$, not the absolute quality of that best fit across models.
Thus any ``anchored Bayes factor'' $\tilde Z_1(x)/\tilde Z_2(x)$ compares curvature/regularization tradeoffs rather than evidence in the
belief-posterior sense.  Model comparison in the decision-posterior regime must instead be aligned with the chosen loss/scoring
rule through predictive assessment.
Any shift-invariant functional of the generalized posterior mapping must factor out data-only terms, and therefore cannot, by itself, retain the cross-model fit component required for Bayes-factor-style evidence.

\subsection{A constructive alternative: shift-invariant predictive scoring}\label{sec:predscore}

Corollary~\ref{cor:gbf-not-well-defined} rules out treating the Gibbs normalizing constant as canonical
belief evidence outside the log-loss (belief-posterior) regime.  A coherent alternative is to compare
procedures by \emph{predictive performance} using a proper scoring rule applied to the predictive
distribution induced by the update.

Let $S(P,y)$ be a (strictly) proper scoring rule for outcomes $y\in\mathcal Y$ and predictive distributions
$P$ on $\mathcal Y$.  To keep notation consistent with the rest of the paper, write $x_t$ for the observed
data at time $t$ and let $y_t=y(x_t)$ denote the component being predicted (possibly $y_t=x_t$).  For
model/procedure $m$, let $q_{m,t}(\cdot)=q_m(\cdot|x_{1:t-1})$ denote the decision posterior after
observing $x_{1:t-1}$, and define the induced one-step-ahead predictive distribution
\begin{equation}\label{eq:pred_induced}
P_{m,t}(\,\cdot|x_{1:t-1})
=
\int p_m(\,\cdot|\theta)\,q_{m,t}(d\theta),
\end{equation}
where $\{p_m(\cdot|\theta):\theta\in\Theta_m\}$ is a predictive family indexed by $\theta$.
(Equivalently, $\theta$ may index \emph{actions} that themselves correspond to predictive distributions; no
joint sampling model for $x_{1:n}$ is required.)  The prequential predictive score is
\begin{equation}\label{eq:prequential_score}
\mathrm{Score}_m(x_{1:n})
=
\sum_{t=1}^n S\!\left(P_{m,t}(\cdot|x_{1:t-1}),\, y_t\right),
\end{equation}
with the analogous held-out score obtained by restricting the sum to a test set.

\begin{proposition}[Shift invariance of predictive scoring]\label{prop:score_shift_invariant}
Suppose the update rule for $q_m(\cdot|x)$ is loss-based and invariant to data-only shifts: for any
measurable $c_m:\mathcal X\to\mathbb R$, replacing $\ell_m(\theta;x)$ by
$\ell_m'(\theta;x)=\ell_m(\theta;x)+c_m(x)$ leaves $q_m(\cdot|x)$ unchanged (hence also $q_{m,t}$ for each $t$).
Then, the induced predictive distributions \eqref{eq:pred_induced} and the predictive score
$\mathrm{Score}_m(x_{1:n})$ in \eqref{eq:prequential_score} are invariant to such shifts.  Consequently,
predictive-score model comparisons are canonical with respect to the shift equivalence
$\ell_m\sim \ell_m+c_m(x)$, unlike comparisons based on $Z_m(x)$.
\end{proposition}

\begin{proof}
Shift invariance implies $q_{m,t}'=q_{m,t}$ for each $t$.  Therefore
\[
P_{m,t}'(\cdot|x_{1:t-1})
=
\int p_m(\cdot|\theta)\,q_{m,t}'(d\theta)
=
\int p_m(\cdot|\theta)\,q_{m,t}(d\theta)
=
P_{m,t}(\cdot|x_{1:t-1}),
\]
and each summand in \eqref{eq:prequential_score} is unchanged.  Summing over $t$ yields
$\mathrm{Score}_m'(x_{1:n})=\mathrm{Score}_m(x_{1:n})$.
\end{proof}

Taking $S(P,y)=-\log p(y)$ yields the cumulative log predictive density
\[
\mathrm{LPD}_m(x_{1:n})
=
\sum_{t=1}^n \log p_{m,t}(y_t|x_{1:t-1}),
\qquad
p_{m,t}(y|x_{1:t-1})=\frac{d}{dy}P_{m,t}(y|x_{1:t-1}),
\]
and model comparisons via $\Delta\mathrm{LPD}_{10}=\mathrm{LPD}_1-\mathrm{LPD}_0$ inherit the shift invariance in
Proposition~\ref{prop:score_shift_invariant}.  For continuous outcomes one may instead use proper scores, such
as continuous ranked probability score (CRPS), to obtain robustness and interpretability on the outcome scale.

\subsection{When does the problem disappear?}\label{subsec:ml-when-ok}

\paragraph{Log-loss/belief posterior case.}
If the loss is negative log-likelihood,
Theorem~\ref{thm:bayes-iff} ensures an ordinary Bayesian posterior
under a normalized sampling model, and Bayes factors regain their standard interpretation.

\paragraph{Decision posterior case.}
Outside log-loss, the natural replacement for belief evidence is
decision-relevant scoring of predictive performance-- prequential or cross-validation criteria coherent for the chosen loss or scoring
rule.  These approaches make explicit
that evidence is defined relative to the decision criterion, rather than inherited
from a generative likelihood.

\section{Expected utility over mixed decision rules implies degeneracy}\label{sec:vnm}

A tempting interpretation of loss-based posteriors is that they represent optimal
randomization.  This
section records a standard obstacle.  If one applies the classical
von~Neumann--Morgenstern expected-utility axioms directly to mixtures of decision rules, then
optimizing over randomized rules collapses to selecting (essentially) a single action.  The
implication: any non-degenerate loss-based posterior must be justified
by nonlinear preferences over decision rules.

\subsection{Preferences over randomized decision rules}\label{subsec:vnm-setup}

Fix observed data $x$ and an action space $\Theta$. A
(randomized) decision rule is a probability measure $q\in\Delta(\Theta)$.  For $q,r\in\Delta(\Theta)$
and $\alpha\in(0,1)$, write $\alpha q + (1-\alpha)r$ for the mixture.  Let $\succeq_x$ denote the decision maker's preference relation over
$\Delta(\Theta)$ after observing $x$.

\paragraph{Axioms (vNM over decision rules).}
We say that $\succeq_x$ satisfies the vNM axioms on $\Delta(\Theta)$ if it is
(i) complete and transitive (weak order);
(ii) continuous in mixtures; and
(iii) satisfies independence: for all $q,r,s\in\Delta(\Theta)$ and all $\alpha\in(0,1)$,
$q \succeq_x r \Leftrightarrow \alpha q+(1-\alpha)s \succeq_x \alpha r+(1-\alpha)s$
\citep{vonNeumannMorgenstern1944TheoryGames}.

\subsection{Degeneracy under vNM independence}\label{subsec:vnm-degeneracy}

\begin{proposition}[vNM over mixtures of decision rules implies degeneracy]\label{prop:vnm-degenerate}
Fix $x$. Suppose $\succeq_x$ satisfies the vNM axioms on $\Delta(\Theta)$.
Then, there exists a measurable function $u_x:\Theta\to\mathbb{R}$, such that preferences admit the expected-utility representation
$U_x(q) = \int_{\Theta} u_x(\theta) q(d\theta)$,
$q \succeq_x r \Leftrightarrow U_x(q)\ge U_x(r)$.
Moreover, if $q^{*}_x\in\arg\max_{q\in\Delta(\Theta)} U_x(q)$ exists, then every maximizer is
supported on
$\Theta_x^{*} = \arg\max_{\theta\in\Theta} u_x(\theta)$,
and in particular there exists an optimal degenerate rule $\delta_{\theta_x^{*}}$. If $\theta_x^{*}$ is unique, the unique optimal
rule is the point mass.
\end{proposition}

Proposition~\ref{prop:vnm-degenerate} makes precise a basic point: under vNM independence applied to
mixtures of decision rules, randomization is behaviorally irrelevant except as tie-breaking
among equally best pure actions.

\begin{corollary}[Non-degenerate loss-based posteriors require nonlinear preferences]\label{cor:nonlinear}
Consider any loss-based update that outputs, for typical data $x$, a non-degenerate distribution
$q(\cdot| x)$ whose support is not contained in an argmax set of a single utility function
$u_x(\theta)$.  Then, $q(\cdot| x)$ cannot be obtained as the unique optimizer of a vNM
expected-utility preference over $\Delta(\Theta)$.
\end{corollary}

Loss-based posteriors become decision-theoretically natural once one allows nonlinear
preferences over randomized rules, for example, by penalizing deviations from a baseline rule,
$\pi$.  Such preferences introduce curvature, making interior solutions optimal and turning
randomization into a meaningful regularization device. 

\section{A coherence book and implications}\label{sec:book}

This section translates the characterizations into a practical
``coherence book.''  When a procedure
yields a belief posterior (Theorem~\ref{thm:bayes-iff}), it inherits the familiar Bayesian
coherence guarantees.  When it yields a decision
posterior (Corollary~\ref{cor:gb}), it is coherent as an optimal randomized decision
rule, but several belief-based corollaries are no longer automatic.
Table~\ref{tab:book} summarizes the key distinctions.

\begin{table}[t]
\centering
\caption{Coherence book: belief posteriors versus decision posteriors}
\label{tab:book}
\small
\setlength{\tabcolsep}{4pt}
\renewcommand{\arraystretch}{1.15}
\begin{tabularx}{\textwidth}{@{}>{\raggedright\arraybackslash}p{0.20\textwidth}
>{\raggedright\arraybackslash}X
>{\raggedright\arraybackslash}X
>{\raggedright\arraybackslash}p{0.22\textwidth}@{}}
\toprule
Property/claim & Belief posterior & Decision posterior & Replacement \\
\midrule
Interpretation of probabilities &
Conditional beliefs &
Weights in randomized decision rule &
Decision-theoretic optimality \\
Existence of a joint model &
Yes (by construction) &
No (only in the log-loss case) &
Loss as primitive \\
Dynamic consistency  &
Conditional choice consistency &
Sequential coherence &
Operator coherence \\
Exchangeability (de Finetti) &
Belief constraint on sequences &
Not implied &
Permutation invariance \\
Likelihood principle &
Available within a specified model &
Not generally applicable &
Loss principle \\
Marginal likelihood/Bayes factor &
Well-defined, model-based &
No canonical analogue in general &
Decision-based model comparison \\
Scale identification &
Likelihood fixes scale &
Scale free &
Calibration/robustness choice for $\eta$ \\
\bottomrule
\end{tabularx}
\end{table}

\subsection{Key implications}\label{subsec:decisions}

The book can be summarized: belief posteriors support belief-based decision
analysis; decision posteriors support rule-based decision analysis.

First, ``prior'' means ``baseline rule'' unless log-loss holds.
When Theorem~\ref{thm:bayes-iff} fails, $\pi$ should not be interpreted as a prior belief.  It is a baseline or reference rule that anchors regularization.  Arguments that rely on prior beliefs-- exchangeability or subjective probability calibration-- do not automatically transfer.  What does transfer is the interpretation of $\pi$ as a complexity control, default action, or initial commitment.

Second, dynamic consistency is replaced, not preserved.
In belief-based Bayes, dynamic consistency links conditional and unconditional preferences
over acts; the updating rule is derived from plans remaining optimal after
conditioning \citep{Ghirardato2002RevisitingSavage}.  In loss-based updating, the coherence notion
is operator coherence: updating sequentially with additive losses
yields the same decision posterior as updating once with the combined loss \citep{BissiriHolmesWalker2016GeneralFramework}.  These are different requirements.

Third, exchangeability is not inherited from order invariance.
Permutation invariance of additive loss does
not imply exchangeability as a belief constraint \citep{deFinetti1937LaPrevision}, because there need not
exist an underlying joint probability model.  Permutation invariance of a criterion is not a substitute for
probabilistic exchangeability.

Fourth, evidential measures become task-dependent.
Belief posteriors come with canonical evidential summaries (marginal likelihood, Bayes factors)
tied to the joint model.  Decision posteriors do not yield canonical model evidence; comparisons
must be formulated in the same decision language that defines the update.

In terms of reporting standards for generalized posteriors, we recommend the following:
to avoid importing belief-based claims into the decision-posterior regime, applied use of loss-based updating should report:
(i) the loss $\ell(\theta;x)$ and its scale/units;
(ii) the learning rate/temperature $\eta$ and the procedure used to select it;
(iii) whether the output is interpreted as a belief posterior (and if so, the implied likelihood model per Theorem~1) or as a decision posterior;
(iv) the intended independence/separability structure (if any) motivating Axiom~A4;
(v) the model-comparison criterion used (preferably predictive scoring aligned with the loss/scoring rule, rather than $Z(x)$);
and (vi) any implementation conventions (centering/normalization) that leave $q(\cdot|x)$ unchanged but affect $Z(x)$.

\subsection{Learning-rate calibration}\label{subsec:eta}

Perhaps the most visible practical consequence is the
status of the learning rate $\eta$.  In ordinary Bayes, the scale of the log-likelihood
is fixed by the probabilistic model, so the natural temperature is one.  In loss-based updating,
$\eta$ is part of the preference specification: it controls concentration and must be
calibrated.

Several principled calibration strategies have been proposed.  The loss-likelihood bootstrap
\citep{LyddonHolmesWalker2019LossLikelihoodBootstrap} selects $\eta$ so that
posterior uncertainty reflects sampling variability of a loss-based estimator.  Other approaches
seek safe learning rates under misspecification \citep{GrunwaldVanOmmen2017InconsistencyMisspecifiedLinearModels}.  The decision-theoretic view makes clear why calibration is central:
it replaces the probabilistic scale identification that is automatic under belief
posteriors.

\subsection{A short checklist}\label{subsec:checklist}

\begin{enumerate}[(C1)]
\item \emph{Interpretation:} Are you claiming belief semantics, or is the posterior a decision rule?
If belief semantics are claimed, Theorem~\ref{thm:bayes-iff} identifies the required loss structure.
\item \emph{Loss justification:} Why does the chosen loss represent the inferential target or
decision problem at hand?
\item \emph{Baseline justification:} What does $\pi$ represent, and how sensitive are conclusions to this choice?
\item \emph{Calibration:} How is the learning rate chosen, and what property is it designed to
deliver (coverage, robustness, predictive safety)?
\end{enumerate}

Loss-based updating is coherent, but its coherence is that of decision posteriors unless the
loss is log-likelihood.  The characterization results of Section~\ref{sec:characterize} make this
distinction sharp; the book translates it into operational guidance.

\section{Discussion}\label{sec:discussion}

Generalized Bayes and related loss-based updates have become standard tools.  Their
success is empirical and methodological: they often deliver stable procedures in settings where a
fully specified likelihood is unavailable, unreliable, or misaligned with the inferential target.
The goal of this paper has been to make explicit the boundary between belief posteriors and decision posteriors.

Generalized Bayes is best understood as a coherent mechanism for selecting randomized decision rules under an explicit
loss and baseline, with Bayesian belief updating recovered only in the log-loss case.  Making this
distinction explicit strengthens the methodological foundations of loss-based inference.

\bibliographystyle{abbrvnat}
\bibliography{reference}

\newpage
\setcounter{equation}{0}
\setcounter{section}{0}
\setcounter{table}{0}
\setcounter{page}{1}
\renewcommand{\thesection}{S\arabic{section}}
\renewcommand{\theequation}{S\arabic{equation}}
\renewcommand{\thetable}{S\arabic{table}}

\vspace{1cm}
\begin{center}
{\LARGE
{\bf Supplementary Material for ``When Is Generalized Bayes Bayesian?"}
}
\end{center}

\section{Proofs and supporting results}\label{app:proofs}

This appendix collects proofs.  Throughout, $\Delta(\Theta)$
denotes probability measures on the measurable space $(\Theta,\mathcal T)$.

\subsection{Proof of Theorem~\ref{thm:bayes-iff}}\label{app:bayes-iff}

\begin{proof}
Assume (i): there exists $\{p_\theta(x)\}$ such that
$q(\theta| x)\propto \pi(\theta)p_\theta(x)$ for all $x$. Comparing with
$q(\theta| x)\propto \pi(\theta)\exp\{-\eta\ell(\theta;x)\}$ implies that for each fixed $x$,
\[
p_\theta(x) = a(x)\exp\{-\eta\ell(\theta;x)\}
\]
for some factor $a(x)>0$ not depending on $\theta$.
Taking logs gives
$\ell(\theta;x)=-(1/\eta)\log p_\theta(x)+c(x)$ with $c(x)=-(1/\eta)\log a(x)$, proving (ii).

Conversely, assume (ii). Then
$\exp\{-\eta\ell(\theta;x)\}=p_\theta(x)\exp\{-\eta c(x)\}$. The factor $\exp\{-\eta c(x)\}$ does not
depend on $\theta$ and cancels in normalization over $\theta$, so the generalized update equals the
Bayesian posterior under likelihood $p_\theta(x)$, proving (i).
\end{proof}

\subsection{From KL-regularized choice to exponential tilting}\label{app:kl-to-gibbs}

\begin{lemma}[Gibbs rule solves KL-regularized risk]\label{lem:gibbs-solves-kl}
Fix $x$, loss $\ell(\theta;x)$, baseline $\pi$, and $\eta>0$. Suppose
$\int \exp\{-\eta\ell(\theta;x)\} \pi(d\theta) < \infty$.
Then the optimization problem
\[
\inf_{q\in\Delta(\Theta)}
\left\{
\int \ell(\theta;x) q(d\theta)+\frac{1}{\eta}D_{\mathrm{KL}}(q\|\pi)
\right\}
\]
has the unique minimizer
\[
q^{*}(d\theta| x)=\frac{\exp\{-\eta\ell(\theta;x)\}\pi(d\theta)}{\int \exp\{-\eta\ell(\vartheta;x)\} \pi(d\vartheta)}.
\]
Moreover, the optimal value equals $-(1/\eta)\log\int \exp\{-\eta\ell(\theta;x)\} \pi(d\theta)$.
\end{lemma}

\begin{proof}
Let $Z=\int \exp\{-\eta\ell(\theta;x)\} \pi(d\theta)$ and define $q^{*}$ as above.
For any $q\ll\pi$, write the KL divergence as
$D_{\mathrm{KL}}(q\|\pi)=\int \log(dq/d\pi) dq$.
Consider
\[
\int \ell\,dq+\frac{1}{\eta}D_{\mathrm{KL}}(q\|\pi)
=\frac{1}{\eta}\int \log\!\left(\exp\{\eta\ell\}\frac{dq}{d\pi}\right) dq.
\]
Add and subtract $\log Z$ and use the definition of $q^{*}$:
\[
\log\!\left(\exp\{\eta\ell(\theta;x)\}\frac{dq}{d\pi}\right)
=\log\!\left(\frac{dq}{dq^{*}}\right)+\log Z.
\]
Integrating under $q$ gives
\[
\int \ell\,dq+\frac{1}{\eta}D_{\mathrm{KL}}(q\|\pi)
=\frac{1}{\eta}D_{\mathrm{KL}}(q\|q^{*})+\frac{1}{\eta}\log Z.
\]
Since $D_{\mathrm{KL}}(q\|q^{*})\ge 0$ with equality iff $q=q^{*}$ (a.s.), $q^{*}$ is the unique
minimizer.
\end{proof}

\subsection{Proofs for Section~\ref{subsec:variational}}\label{app:variational}

The argument in Section~\ref{subsec:variational} has three logically distinct components:
(i) a variational representation of preferences over randomized decision rules,
(ii) a product-additivity restriction on the deviation cost implied by separability across independent
subproblems, and (iii) identification of the unique product-additive deviation cost within the
$f$-divergence class, which singles out relative entropy (up to scale).

\paragraph{Step 1: variational representation (assumed).}
We isolate the role of the deviation penalty through
Assumption~\ref{ass:variational}, which postulates that conditional preferences over randomized rules
admit the regularized form \eqref{eq:variational-form}.  This assumption is standard in the decision-theoretic literature
\citep{MaccheroniMarinacciRustichini2006,Strzalecki2011}.

\paragraph{Step 2: separability implies product additivity.}

\begin{proof}[Proof of Proposition~\ref{prop:prodadd}]
Consider a factorized decision problem with $\Theta=\Theta_1\times\Theta_2$,
$\pi=\pi_1\otimes \pi_2$, and additive loss.
Under Assumption~\ref{ass:variational}, 
Axiom (A4) requires that comparing $q_1\otimes q_2$
and $r_1\otimes q_2$ be
independent of the specific choice of $q_2$.  Applied to the
variational criterion, this implies that the incremental deviation cost
\[
\Delta(q_1,r_1;q_2)
:=
D(q_1\otimes q_2\|\pi_1\otimes\pi_2)-D(r_1\otimes q_2\|\pi_1\otimes\pi_2)
\]
cannot depend on $q_2$.  Setting $r_1=\pi_1$ yields that
$D(q_1\otimes q_2\|\pi_1\otimes\pi_2)-D(\pi_1\otimes q_2\|\pi_1\otimes\pi_2)$ depends only on $q_1$.
By symmetry, swapping the roles of the components gives an analogous statement depending only on
$q_2$.  Combining these two invariances and using $D(\pi_1\otimes \pi_2\|\pi_1\otimes\pi_2)=0$ yields
the decomposition
\[
D(q_1\otimes q_2\|\pi_1\otimes\pi_2)
=
D(q_1\otimes \pi_2\|\pi_1\otimes\pi_2)
+
D(\pi_1\otimes q_2\|\pi_1\otimes\pi_2).
\]
Interpreting $D(\cdot\|\cdot)$ consistently across domains gives
$D(q_1\otimes \pi_2\|\pi_1\otimes\pi_2)=D(q_1\|\pi_1)$ and
$D(\pi_1\otimes q_2\|\pi_1\otimes\pi_2)=D(q_2\|\pi_2)$, hence product additivity.
\end{proof}

\paragraph{Step 3: within $f$-divergences, product additivity forces KL.}

\begin{proof}[Proof of Proposition~\ref{prop:klunique}]
Under Assumption~\ref{ass:fdiv}, for $q\ll \pi$ we can write
$D(q\|\pi)=\int \varphi(dq/d\pi)\,d\pi$ with $\varphi$ convex and $\varphi(1)=0$.
For product measures $q=q_1\otimes q_2$ and $\pi=\pi_1\otimes\pi_2$, the Radon--Nikodym derivative
factorizes.
Hence product additivity from Proposition~\ref{prop:prodadd} is equivalent to the identity
\[
\int \varphi\!\left( A(\theta_1)B(\theta_2)\right)\,(\pi_1\otimes\pi_2)(d\theta_1\,d\theta_2)
=
\int \varphi(A(\theta_1))\,\pi_1(d\theta_1)
+
\int \varphi(B(\theta_2))\,\pi_2(d\theta_2),
\]
holding for all measurable $A=dq_1/d\pi_1$ and $B=dq_2/d\pi_2$ with
$\int A\,d\pi_1=\int B\,d\pi_2=1$.

A standard functional-equation argument shows that the convex solutions are
$\varphi(t)=c\,t\log t + a(t-1)$, $c\ge 0$,
and strict convexity forces $c>0$.  The linear term integrates to zero because
$\int a(dq/d\pi-1)\,d\pi = a(\int dq - \int d\pi)=0$, so it does not
affect $D(q\|\pi)$.  Therefore $D(q\|\pi)=c\,D_{\mathrm{KL}}(q\|\pi)$, as claimed.
\end{proof}

\paragraph{Step 4: exponential tilting.}
With $D$ proportional to $D_{\mathrm{KL}}$, the optimizer of \eqref{eq:variational-form} is the Gibbs
rule by Lemma~\ref{lem:gibbs-solves-kl}, and the multiplicative constant is absorbed into $\eta$.

\begin{proof}[Proof of Corollary~\ref{cor:gb}]
Combine Proposition~\ref{prop:prodadd} and Proposition~\ref{prop:klunique} to conclude that
$D(q\|\pi)=cD_{\mathrm{KL}}(q\|\pi)$ with $c>0$.  Absorb $c$ into $\eta$ in
\eqref{eq:variational-form}.  Lemma~\ref{lem:gibbs-solves-kl} then yields the generalized Bayes form
\eqref{eq:gb-update}; uniqueness follows from strict convexity.
\end{proof}

\begin{remark}[What is proved versus assumed]\label{rem:reduction}
Assumption~\ref{ass:variational} isolates the variational structure of preferences
over randomized decision rules.  Given this
structure, Proposition~\ref{prop:prodadd} is a direct consequence of the separability axiom (A4).
Proposition~\ref{prop:klunique} provides a self-contained identification of relative entropy within
the broad class of separable $f$-divergence deviation costs; together with
Lemma~\ref{lem:gibbs-solves-kl}, this yields a compact and fully explicit derivation of the
generalized Bayes decision rule under the stated assumptions.
\end{remark}

\end{document}